\documentclass[runningheads]{llncs}
\usepackage{algorithm}
\usepackage{algorithmic}
\usepackage{subcaption}
\usepackage{float}
\usepackage{graphicx}
\usepackage{xcolor}
\usepackage{hyperref}       
\usepackage{url}
\usepackage{ dsfont }

\usepackage{amsmath, amssymb}
\begin{document}
\title{Positional queuetions}
%
%
\author{Vladimir Yankovskiy\inst{1, 2}\orcidID{0000-0002-6013-6319}}
\authorrunning{V. Yankovskiy}
%
\institute{High School of Economics, Doctoral School of Mathematics \and Yandex \email{vladimir.vank0vskiy@yandex.com}}
\maketitle              
\begin{abstract}
In this work, we consider properties of VCG and GSP auctions positions in queues.

\end{abstract}
\section{Introduction}

In recent years, online marketplaces have become increasingly popular as a platform for buying and selling goods. With the rise of marketplaces, there has been a growing interest in sponsored search and position auctions as a means of generating revenue. Sponsored search is a form of advertising where advertisers bid for the right to have their ads displayed alongside search results. Position auctions, on the other hand, determine the order in which ads are displayed in search results. 

In his seminal paper \cite{varian2007position} Hal Varian discusses the use of position auctions in sponsored search and their effectiveness in generating revenue. Position auctions have been extensively studied in the context of search engines (\cite{nisan2007algorithmic}, \cite{roughgarden2016twenty}, \cite{athey2011position}, \cite{varian2009online}, \cite{ashlagi2010position}, \cite{edelman2010optimal}, \cite{aggarwal2008sponsored}, \cite{milgrom2007simplified}, \cite{borgers2013equilibrium}, \cite{jansen2008sponsored}, \cite{deng2021towards}, \cite{bayir2019genie}), and a bit less extensively in the context of marketplaces (\cite{li2023optimally}, \cite{watts2021fairness}, \cite{ferreira2022learning}, \cite{derakhshan2022product}, \cite{dash2021umpire}, \cite{etro2021product}, \cite{yankovskiy2024position}) 

 In our paper a position auction fitted to a completely different task: the task of auctioning positions in a queue. 

For it we get upper and lower bounds on possible revenue that can be estimated beforehand. The main result of this work is such bounds for VCG and GSP auctions.

Our analysis is done in the same way as in our previous article \cite{yankovskiy2024position}.

\section{General setting of position queuetions}

Suppose we have $N$ participants waiting in a queue for some service. To each participant there correspond two numbers: the time which it takes to serve them $t_i > 0$ and their value for time $w_i > 0$ (known to the participants but not to the organizer). 

If the participants are distributed in a queue according to some $\sigma(i)$, then $\sigma(i)$-th participant will gain $-w_\sigma(i) \Sigma_{j < i} t_\sigma(j)$ (that is their value of the time they will wait in queue).

Within this setting we will consider two mechanisms and the games that correspond to them. First, let's recall the definition of a mechanism.

\begin{definition}[\cite{nisan2007algorithmic}]
A mechanism for $n$ players is given by 
\begin{itemize}
    \item players type spaces $T_1, ..., T_N$
    \item players action spaces $X_1, ..., X_N$
    \item alternative set $A$
    \item players valuation functions $V_i: T_i \times A \to \mathbb{R}$
    \item outcome function $a: X_1 \times ... \times X_N \to A$
    \item players payment functions $p_i: X_1, ..., X_N \mathbb{R}$
\end{itemize}
The game induced by the mechanism is given by using $T_i$ as type spaces, $X_i$ as action spaces and utilities
$$u_i(t_i, x_1, ..., x_N) = V_i(t_i, a(x_1, ..., x_N)) - p_i(x_1, ..., x_N)$$
\end{definition}

\section{Optimal sorting of customers}

\begin{lemma}
The total loss of the participants reaches its minimum when the participants ars sorted in non-ascending order by
$$v_i := \frac{w_i}{t_i}$$
\end{lemma}

\begin{proof}
First, notice that the total loss of the participants is $\Sigma_{j < i} v_{\sigma(i)} t_{\sigma(j)}$.

Suppose $v_{\sigma(i)} < v_{\sigma(i+1)}$. Then $w_{\sigma(i + 1)} t_{\sigma(i)} > w_{\sigma(i)} t_{\sigma(i + 1)}$. And swapping the places of the $\Sigma(i)$-th and the $\Sigma(i + 1)$ customers will increase the overall sum.
\end{proof}

\section{VCG queuetions}

\begin{definition}[VCG queuetion]
The participants place bids $b_j$ and then are given places $\sigma^{-1}(j)$ so, that $b_{\sigma(i)}$ is non-increasing. 
Then $j$-th participant pays $$t_\sigma(i) \Sigma_{j > i} t_{\sigma(j)} b_{\sigma(j)}$$

That is exactly what all other participants lose from their presence.
\end{definition}

The total loss of the participant on $j$-th slot is $$t_\sigma(i) (v_\sigma(i) \Sigma_{j < i} t_\sigma(j) + \Sigma_{j > i} t_{\sigma(j)} b_{\sigma(j)})$$ and the profit of the organizer is $$\Sigma_{j > i} t_\sigma(i) t_{\sigma(j)} b_{\sigma(j)}$$

\begin{lemma}
The combination of bids forms a Nash Equilibrium for VCG auction if and only if

$$v_{\sigma(k+1)} \leq b_{\sigma(k)} \leq v_{\sigma(k-1)}$$ 
\end{lemma}

\begin{proof}
    One can notice, that if the participant on $k$-th slot changes their bid to move to the $j$-th slot, then their profit will change to $$t_{\sigma(k)} (v_{\sigma(k)} \Sigma_{i < j} t_{\sigma(i)} + \Sigma_{i \geq j; i \neq k} t_{\sigma(i)} b_{\sigma(i)})$$ in case $j < k$ or to $t_{\sigma(k)} (v_\sigma(k) \Sigma_{i \leq j; i \neq k} t_\sigma(i) + \Sigma_{i > j} t_{\sigma(i)} b_{\sigma(i)})$ otherwise.

    Therefore, a sequence of bids $b_1, …, b_N$ forms a Nash equilibrium if and only if

    $$v_\sigma(k) \Sigma_{j \leq i < k} t_\sigma(i) \leq \Sigma_{j \leq i < k} t_{\sigma(i)} b_{\sigma(i)}$$

    for all $j < k$ and 

    $$v_\sigma(k) \Sigma_{k < i \leq j} t_\sigma(i) \geq \Sigma_{k < i \leq j} t_{\sigma(i)} b_{\sigma(i)}$$

    for $k < j$

    From this it follows that for every $k$ it holds that

    $$max_{j} \frac{\Sigma_{j \leq i < k} t_{\sigma(i)} b_{\sigma(i)}}{\Sigma_{j \leq i < j} t_\sigma(i)} \leq v_\sigma(k) \leq min_{j} \frac{\Sigma_{k < i \leq j} t_{\sigma(i)} b_{\sigma(i)}}{\Sigma_{k < i \leq j} t_\sigma(i)}$$

    In particular, 

    $$b_{\sigma(k+1)} \leq v_\sigma(k) \leq b_{\sigma(k-1)}$$

    On the other hand as $$b_{\sigma(j)} \geq \frac{\Sigma_{j \leq i < k} t_{\sigma(i)} b_{\sigma(i)}}{\sigma_{j \leq i < j} t_\sigma(i)} \leq b_{\sigma(k-1)}$$ due to the bids being well-ordered, we see that the condition

    $$b_{\sigma(k+1)} \leq v_\sigma(k) \leq b_{\sigma(k-1)}$$

    is sufficient.
\end{proof}

\begin{theorem}
In any Nash equilibrium for VCG $v_{\sigma(i)} \geq v_{\sigma(j)}$ for all $j > i+1$
\end{theorem}

\begin{proof}
    $$v_{\sigma(i)} \geq b_{\sigma(i+1)} \geq b_{\sigma(j-1)} \geq v_{\sigma(j)}$$
\end{proof}

This allows to calculate the exact lower bound for the organizers revenue in $O(N^2)$ via dynamic programming. For the exact upper bound we have an exact formula, which can be calculated in $O(N)$:

\begin{theorem}
    The maximum possible organizer's revenue in a Nash equilibrium in $VCG$ is equal to $$\Sigma_{j > i} t_\sigma(i) t_{\sigma(j)} v_{\sigma(j)}$$.
\end{theorem}

\begin{proof}
Let's construct the maximum possible combination of bids that forms a Nash equilibrium.

By Lemma 1 for a fixed $\sigma$ it is achieved exactly when 

$$b_{\sigma(k)} = v_{\sigma(k - 1)}$$

Because the organizer's revenue monothonously depends on bids, their revenue with this set of bids will be maximal. And it will be equal to

$$\Sigma_{j > i} t_\sigma(i) t_{\sigma(j)} v_{\sigma(j)}$$
\end{proof}

\section{Generalized second price auctions with preference}

Consider yet another type of rules:

\begin{definition}[GSP positional auction with preference]
The participants place bids $B_j$ and then are given places $\sigma^{-1}(j)$ so, that $B_{\sigma(i)}$ is non-increasing. 
Then $i$-th participant pays $B_{\sigma(i+1)} t_i$.
\end{definition}

The profit of $j$-th participant will be $$t_\sigma(i) (v_\sigma(i) \Sigma_{j < i} t_\sigma(j) + B_i$$ and the profit of the organizer will be $\sum_{i = 1}^N B_{\sigma(i + 1)} t_i$ — the same as in VCG.

\begin{lemma}
The combination of bids forms a Nash Equilibrium for GSP auction in and only if for all $j < k$ holds:

$$v_{\sigma(k-1)}(\Sigma_{j \leq i < k-1} t_\sigma(i)) \leq B_{\sigma(j)} - B_{\sigma(k)} \leq v_{\sigma(k-1)} (\Sigma_{j-1 \leq i < k-1} t_\sigma(i))$$

\end{lemma}

\begin{proof}
One can notice, that if the participant on $k$-th slot changes their bid to move to the $j$-th slot, then their profit will change to $$t_{\sigma(k)} (v_{\sigma(k)} \Sigma_{i < j} t_{\sigma(i)} + \Sigma_{i \geq j; i \neq k} t_{\sigma(i)} b_{\sigma(i)})$$ in case $j < k$ or to $t_{\sigma(k)} (v_\sigma(k) \Sigma_{i \leq j; i \neq k} t_\sigma(i) + \Sigma_{i > j} t_{\sigma(i)} b_{\sigma(i)})$ otherwise.

    Therefore, a sequence of bids $b_1, …, b_N$ forms a Nash equilibrium if and only if

    $$v_\sigma(k) \Sigma_{j \leq i < k} t_\sigma(i) \leq \Sigma_{j \leq i < k} t_{\sigma(i)} b_{\sigma(i)}$$

    for all $j < k$ and 

    $$v_\sigma(k) \Sigma_{k < i \leq j} t_\sigma(i) \geq \Sigma_{k < i \leq j} t_{\sigma(i)} b_{\sigma(i)}$$

    for $k < j$

    From this it follows that for every $k$ it holds that

    One can notice, that if the participant on $k$-th slot changes their bid to move to the $j$-th slot, then their profit will change to $$t_{\sigma(k)} (v_{\sigma(k)} \Sigma_{i < j} t_{\sigma(i)} + \Sigma_{i \geq j; i \neq k} t_{\sigma(i)} b_{\sigma(i)})$$ in case $j < k$ or to $t_{\sigma(k)} (v_\sigma(k) \Sigma_{i \leq j; i \neq k} t_\sigma(i) + \Sigma_{i > j} t_{\sigma(i)} b_{\sigma(i)})$ otherwise.

    Therefore, a sequence of bids $b_1, …, b_N$ forms a Nash equilibrium if and only if

    $$v_\sigma(k) \Sigma_{j \leq i < k} t_\sigma(i) \leq \Sigma_{j \leq i < k} t_{\sigma(i)} b_{\sigma(i)}$$

    for all $j < k$ and 

    $$v_\sigma(k) \Sigma_{k < i \leq j} t_\sigma(i) \geq \Sigma_{k < i \leq j} t_{\sigma(i)} b_{\sigma(i)}$$

    for $k < j$

    From this it follows that for every $k$ it holds that

     $$max_{j} \frac{B_{\sigma(j)} - B_{\sigma(k + 1)}}{\Sigma_{j \leq i < k} t_\sigma(i)} \leq v_\sigma(k) \leq min_{j} \frac{B_{\sigma(k + 1)} - B_{\sigma(j + 1)}}{\Sigma_{k < i \leq j} t_\sigma(i)}$$

     That means, for all $j < k$ we have 
     $$v_{\sigma(k-1)}(\Sigma_{j \leq i < k-1} t_\sigma(i)) \leq B_{\sigma(j)} - B_{\sigma(k)} \leq v_{\sigma(k-1)} (\Sigma_{j-1 \leq i < k-1} t_\sigma(i))$$

\end{proof}

\begin{theorem}
The maximal possible revenue of the organizer in case when all $v_{\sigma(i)}$-s are sorted in non-descending order equals to equals to $\sum_{j \leq i} t_{\sigma(j)} t_\sigma(i) v_{\sigma(i)}$ 
\end{theorem}

\begin{proof}
     $$B_{\sigma(j)} - B_{\sigma(k)} \leq v_{\sigma(k-1)} (\Sigma_{j-1 \leq i < k-1} t_\sigma(i)) \leq (\Sigma_{j-1 \leq i < k-1} t_\sigma(i) v_{\sigma(k-1)})$$

     From that it follows that $B_{\sigma(j)} \leq (\Sigma_{j-1 \leq i} t_\sigma(i) v_{\sigma(k-1)})$, and that $B_{\sigma(j)} = (\Sigma_{j-1 \leq i} t_\sigma(i) v_{\sigma(k-1)})$ satisfies the Nash Equilibrium conditions.

     Because the organizer's revenue monothonously depends on bids, their revenue with this set of bids will be maximal. And it will be equal to

     $$\Sigma_{j > i} t_\sigma(i) t_{\sigma(j)} v_{\sigma(j)}$$
\end{proof}

%
%
\bibliographystyle{splncs04}
\bibliography{mybibliography}

\end{document}